%% file: main-arxiv.tex
\documentclass[conference]{IEEEtran}
\IEEEoverridecommandlockouts

\usepackage{cite}
\usepackage{amsmath,amssymb,amsfonts}
\usepackage{algorithmic}
\usepackage{graphicx}
\usepackage{textcomp}
\usepackage{xcolor}
\usepackage{hyperref,cleveref}
\input{Cheadmin.tex}
\def\BibTeX{{\rm B\kern-.05em{\sc i\kern-.025em b}\kern-.08em
    T\kern-.1667em\lower.7ex\hbox{E}\kern-.125emX}}
\begin{document}

\title{Double Markovity for quantum systems\thanks{M.H. was supported in part by 
the Guangdong Provincial Quantum Science Strategic Initiative (Grant No. GDZX2505003), 
the General R\&D Projects of 1+1+1 CUHK-CUHK(SZ)-GDST Joint Collaboration Fund (Grant No. GRDP2025-022), 
and
the Shenzhen International Quantum Academy (Grant No. SIQA2025KFKT07).}
}

\author{
\IEEEauthorblockN{Masahito Hayashi}
\IEEEauthorblockA{\textit{School of Data Science} \\
\textit{The Chinese University of Hong Kong, Shenzhen}, China \\
\textit{International Quantum Academy (SIQA)}, Shenzhen, China \\
\textit{Graduate School of Mathematics, Nagoya University}, Japan \\
email: hmasahito@cuhk.edu.cn}
\and
\IEEEauthorblockN{Jinpei Zhao}
\IEEEauthorblockA{\textit{Department of Information Engineering} \\
\textit{The Chinese University of Hong Kong}\\
Shatin, Hong Kong, China \\
jinpei@ie.cuhk.edu.hk}
}

\maketitle

\begin{abstract}
The subadditivity--doubling--rotation (SDR) technique is a powerful route to Gaussian optimality in classical information theory and relies on strict subadditivity and its equality-case analysis, where double Markovity is a standard tool. We establish quantum analogues of double Markovity. For tripartite states, we characterize the simultaneous Markov conditions A–B–C and A–C–B via compatible projective measurements on B and C that induce a common classical label J yielding A–J–(B,C). For strictly positive four-party states, we show that A–(B,D)–C and A–(C,D)–B hold if and only if A–D–(B,C) holds. These results remove a key bottleneck in extending SDR-type arguments to quantum systems.
\end{abstract}


\section{Introduction}

Gaussian distributions play a distinguished role not only in classical information theory,
but also in quantum information theory, where Gaussian extremality questions arise in a variety
of entropy inequalities and optimization problems for quantum systems.
A recurring theme---spanning functional inequalities and multi-terminal problems---is that
extremizers of information-theoretic functionals under second-moment constraints are often Gaussian.

In the classical setting, a particularly robust route to Gaussian extremality is the
\emph{subadditivity--doubling--rotation} (SDR) proof paradigm.
At the heart of SDR lies a \emph{strict} form of subadditivity: one needs to identify when
equality in a subadditivity inequality forces an independence (or conditional independence)
structure.
A widely used tool for establishing such equality conditions is the \emph{double Markovity}
property of Markov chains (cf.\ \cite[Exercise~16.25]{csk11}), which, informally, says that
if $A$ is conditionally independent of $C$ given $B$ and also conditionally independent of $B$
given $C$, then $B$ and $C$ must share a common ``summary'' variable that screens off $A$.

Historically, the ``doubling--rotation'' idea originates in functional analysis:
it appeared in Lieb's work \cite{lie90} on Gaussian kernels and in Carlen's work \cite{Car91}
on Gaussian optimality for the log-Sobolev inequality.
In information theory, it was brought to the forefront by Geng and Nair \cite{GengNair14}
in the study of Gaussian broadcast channels, where it was combined with subadditivity to
establish Gaussian optimality.
Since then, SDR-type arguments have been used to prove a number of information inequalities
\cite{Cou18,CouJia14,kne15,Gol16,YangEtal17,ZhaEtal18}, and their use in proving the entropy power
inequality was developed in \cite{ajn19}.

While quantum information theory has analogues of many classical inequalities, extending the SDR
paradigm to quantum systems faces a conceptual bottleneck: an appropriate \emph{quantum} counterpart
of double Markovity has not been available in a form suitable for equality-case analysis.
The purpose of this paper is to fill this gap by establishing two quantum versions of double Markovity
that parallel the most commonly used classical formulations. These results remove an important obstacle
in transporting the SDR methodology to quantum information-theoretic settings.

To formulate these quantum analogues precisely, we next recall the notion of a quantum Markov chain
and state the corresponding ``double Markovity'' question in quantum terms.
We work with finite-dimensional quantum systems $A,B,C,\ldots$ with Hilbert spaces
$\mathcal H_A,\mathcal H_B,\mathcal H_C,\ldots$.
A (density) state on $A$ is a positive semidefinite operator $\rho_A$ with $\Tr\rho_A=1$;
for a composite system $AB$ we write $\rho_{AB}$ for a joint state and
$\rho_A=\Tr_B\rho_{AB}$ for the marginal.
To express Markov-type conditional independence in quantum terms, we use the conditional mutual
information (CMI) defined from von Neumann entropies, which are defined as
$S(A):=S(\rho_A):=-\Tr(\rho_A\log\rho_A)$.
We also write the conditional entropy as $S(A|B):=S(AB)-S(B)$.
For a tripartite state $\rho_{ABC}$, the CMI is
\begin{align}
I(A;C|B) := S(AB)+S(BC)-S(B)-S(ABC).
\end{align}
Strong subadditivity asserts $I(A;C|B)\ge 0$ for all $\rho_{ABC}$, and the equality case
$I(A;C|B)=0$ is known to be equivalent to $\rho_{ABC}$ being a \emph{quantum Markov state}
with a specific structural decomposition \cite{HJP}.
Motivated by the classical identity ``$A-B-C$ is Markov $\Leftrightarrow I(A;C|B)=0$'',
we say that $A-B-C$ forms a (quantum) Markov chain if $I(A;C|B)=0$.

The classical double Markovity phenomenon concerns the simultaneous validity of
$A-B-C$ and $A-C-B$.
In the classical case, this implies the existence of a common random variable $J$
that is a function of $B$ and also a function of $C$, such that
$A-J-(B,C)$ is Markov.
This statement is one of the key devices for converting equality in subadditivity-type
inequalities into an independence structure, and hence is tightly intertwined with SDR
arguments.
In the quantum case, however, the non-commutativity of observables means that one cannot
a priori expect a common \emph{function} of $B$ and $C$ in the same sense,
and it is unclear what structural statement replaces it.
Our first main result provides an exact quantum analogue:
we show that ``$A-B-C$ and $A-C-B$'' is equivalent to the existence of compatible
projective measurements on $B$ and $C$ (equivalently, a common block structure)
that induce a classical label $J$ for which
$A-J-(B,C)$ holds.

A second, equally important, variant of double Markovity appears in SDR-based proofs of
entropy power type inequalities, where one works with \emph{conditional} Markov relations
involving an auxiliary variable and requires a full-support assumption to rule out
degenerate cases.
We establish a quantum counterpart of this conditional double Markovity:
for a strictly positive joint state $\rho_{ABCD}$,
the pair of Markov relations $A-(B,D)-C$ and $A-(C,D)-B$ is equivalent to $A-D-(B,C)$.
The proof relies on a refinement of the structural decomposition of quantum Markov states,
including a uniqueness property of a minimal direct-sum decomposition under full support.

The main contributions of this work are as follows.
\begin{itemize}
\item \textbf{Quantum double Markovity.}
For a tripartite state $\rho_{ABC}$, we characterize the simultaneous Markov conditions
$A-B-C$ and $A-C-B$ in terms of a block-diagonal structure induced by PVMs on $B$ and $C$,
and we show that the resulting classical index $J$ renders $A-J-(B,C)$ Markov.
This is stated precisely in Theorem \ref{TH1}.

\item \textbf{Conditional quantum double Markovity under full support.}
For a strictly positive $\rho_{ABCD}$, we prove that
$A-(B,D)-C$ and $A-(C,D)-B$ hold if and only if $A-D-(B,C)$ holds.
This is stated in Theorem \ref{TH2}.
\end{itemize}
To prove the second result, we introduce the minimality of the decomposition introduced in 
\cite{HJP}.

The remainder of this paper is organized as follows.
In Section \ref{sec:SDR}, we review the SDR paradigm and explain how double Markovity enters
equality-case analyses; in particular, we present two classical versions of double Markovity
that are most frequently used in the literature.
Sections \ref{sec:quantumDM} and \ref{sec:quantumDM2} prove their quantum extensions, establishing \Cref{TH1,TH2}, respectively.
We conclude with a brief discussion of implications for extending SDR arguments to
quantum information theory.

\section{The ``subadditivity--doubling-rotation" technique and double Markovity}\label{sec:SDR}
We first recall the classical SDR template and two classical double-Markovity lemmas that will be extended to quantum systems.The idea of the subadditive--doubling rotation technique can be summarized as follows:
Let $f$ be some functional defined on probability distributions. Let $f(X)$ be a \textit{short-hand} notation for $f(p_X)$. Suppose $f$ has the following properties: 
\begin{itemize}
    \item (Strict) subadditivity: $f(X_1,X_2)\leq f(X_1)+f(X_2)$ with equality \textit{if and only if} $X_1\perp X_2$, where $\perp$ is reserved for (classical) independence when discussing classical random variables;
    \item Rotation invariant: $f(QX) = f(X)$, where $Q$ is a unitary matrix. A common running example is $X_+ = \frac{1}{\sqrt{2}}(X_1+X_2), X_- = \frac{1}{\sqrt 2}(X_1-X_2)$.
    \item $f$ is continuous with respect to weak convergence.
\end{itemize} 

The goal is to maximize $f(X)$ subject to $E[XX^T]\preceq K$. Suppose $p_{X^*}$ is a maximizer for $f$ with corresponding maximum value $v$, then take two independent copies of $X^*$, say $X_1$ and $X_2$. \\Note that
\begin{subequations}\label{subadditive}
\begin{align}
    2v &= f(X_1)+f(X_2)= f(X_1,X_2)\\
    &= f(X_+, X_-)\leq f(X_+)+f(X_-)\leq 2v,
\end{align}
\end{subequations}
thus, each inequality must have been an equality in the above derivation. Hence, we have $X_+\perp X_-$. Finally, by the Kac-Berstein's theorem \cite{kac39},\cite{Bernstein1941}, we conclude that $X_1,X_2$ are Gaussian distributed.

Double Markovity plays an important role in verifying the strict subadditivity of the functional $f$. This is because the ``only if'' direction of strict subadditivity is often not obvious. The derivation in (\ref{subadditive}) typically leads only to a certain Markov structure of $(X_+,X_-)$, rather than directly to their independence.

The following two lemmas, stated with proof, are two classical versions of double Markovity. The first is the standard formulation, while the second is the version used in the proof of the entropy power inequality in \cite{ajn19}. Our main contribution is to provide quantum extensions of both versions.

\begin{lemma}
Let \(A,B,C\) be three (real-valued) random variables defined on the same probability space, such that both the Markov chains \(A \to B \to C\) and \(A \to C \to B\) hold. Then  
\begin{enumerate}
\item There exist functions \(g_1(B)\) and \(g_2(C)\) such that \(g_1(B)=g_2(C)\) with probability one.
\item \(A \to g_1(B) \to (B,C)\) is Markov.
\end{enumerate}
\end{lemma}

\begin{proof}
Let \(F_{A\mid B=b}, F_{A\mid C=c}\) denote the (regular) conditional distributions of \(A\) conditioned on \(B=b\) and \(C=c\) respectively. We define an equivalence class according to \(b_1 \equiv b_2\) if the conditional distributions satisfy \(F_{A\mid B = b_1} = F_{A\mid B = b_2}\). This defines a partition of \(\mathcal B\). Denote the partitions by $\mathcal B_i$.

Now define \(g_1(b)\) by $b\mapsto i$ where $b\in \mathcal B_i$. Thus there is a bijection between \(g_1(b)\) and the conditional distributions \(F_{A\mid B=b}\). Let us similarly define an equivalence class and a partition of the values of \(C\) by $\mathcal C = \bigcup_j \mathcal C_j$. 

The Markov chains imply that \(F_{A\mid B,C} = F_{A\mid B}\) with probability one, and \(F_{A\mid B,C} = F_{A\mid C}\) with probability one. Therefore we have \(F_{A\mid B} = F_{A\mid C} = F_{A\mid B,C}\) with probability one. On this set define \(g_2(C)\) to take the same value as \(g_1(B)\). Clearly, by construction \(g_2(C) = g_1(B)\) with probability one. Further it is also clear by construction that \(F_{A\mid B,C,g_1(B)} = F_{A\mid B,C} = F_{A\mid B} = F_{A\mid g_1(B)}\) with probability one. Thus we also have \(A \to g_1(B) \to (B,C)\) is Markov.
\end{proof}

\begin{lemma}
    Let $A,B,C,D$ be (real-valued) random variables defined on the same probability space, such that 
    both Markov chains $A\to (B,D)\to C$ and $A\to (C,D)\to B$ hold. Suppose $p(a,b,c|d)$ has everywhere non-zero density. Then
$        A\to D\to (B,C)$  forms a Markov chain.
\end{lemma}
\begin{proof}
Both Markov chains $A\to (B,D)\to C$ and $A\to (C,D)\to B$
yield the relaiton $p(a|b,d) = p(a|b,c,d) = p(a|c,d)$ for any $a,b,c,d$, which implies
    \begin{align*}
        &p(a|b,c,d) = p(a|b,d)\\
        =& \sum_c p(c)p(a|b,d)
        = \sum_c p(c)p(a|c,d)
        = p(a|d).
    \end{align*}
\end{proof}

\section{Double Markovity for quantum systems}\label{sec:quantumDM}
Let $A,B,C$ be quantum systems with joint state $\rho_{ABC}$. We say that $A- B- C$ forms a Markov chain if $I(A;C|B) = 0$. 

\begin{theorem}\label{TH1}
The following statements are equivalent:
\begin{enumerate}[(i)]
    \item The Markov chains $A-B-C$ and $A-C-B$ hold simultaneously;
    \item There exist PVMs (projective-valued measures) $\{E_{B,j}\}_j$ on $\mathcal{H}_B$ and $\{E_{C,j}\}_j$ on $\mathcal{H}_C$ such that:
  \begin{equation}
    \sum_{j} (I_A \otimes E_{B,j} \otimes E_{C,j})\rho_{ABC}(I_A \otimes E_{B,j} \otimes E_{C,j}) = \rho_{ABC},
    \label{eq:block-diag}
  \end{equation}
  and if we define
  \begin{align}
    \rho_{ABCJ} 
    := &\sum_j \Bigl((I_A \otimes E_{B,j} \otimes E_{C,j})\rho_{ABC}\notag\\
    &(I_A \otimes E_{B,j} \otimes E_{C,j})  
    \otimes |j\rangle\langle j|\Bigr),
    \label{eq:rho-ABCJ}
  \end{align}
  then the Markov chain $A-J-(B,C)$ holds, where $J$ is the classical index from the PVM measurement.
\end{enumerate}
\end{theorem}
\begin{proof}
\noindent{\bf Step 1:}
We first note that $(ii)\implies (i)$ is immediate. Suppose $(ii)$ is true. 
In (ii), the block-diagonal condition with respect to the PVM $\{E_{B,j}\}_j$ implies that
the conditional states $\rho_B^{(j)}$ have pairwise orthogonal supports. Hence $J$ can be
recovered from $B$ with zero error, i.e., $S(J|B)=0$ (and similarly $S(J|C)=0$), and thus
$I(A;C|B)=I(A;C|B,J)=0$.

\noindent{\bf Step 2:}
For the reverse direction ($(i) \implies (ii)$), suppose $A-B-C$ and $A-C-B$ hold. 
We first prove a special case when $B$ and $C$ are classical random variables. 
Note that the state $\rho_{ABC}$ can be written as
$\sum_{b\in {\cal B},c\in{\cal C}}P_{B,C}(b,c) \rho_{A|b,c}\otimes |b,c\rangle \langle b,c|$.

Consider a pair $(b,c)$ such that $P_{B,C}(b,c)>0$. 
The Markov chains $A-B-C$ and $A-C-B$ imply
$\rho_{A|b,c}$ that depends only on either $b$ or $c$.
That is, 
$\rho_{A|B=b}=\rho_{A|C=c}=\rho_{A|b,c}$. Based on this observation, we divide the set 
${\cal B}$ into disjoint subsets ${\cal B}_j$ with $\cup_j  {\cal B}_j={\cal B}$ such that 
$b$ and $b'$ belong to the same subset if and only if
$\rho_{A|B=b}=\rho_{A|B=b'} $.
Set ${\cal C}$ is also divided into disjoint subsets ${\cal C}_j$ in the same way. Additionally, we require the label $j$ for the subsets ${\cal C}_j$ satisfy $\rho_{A|B=b}=\rho_{A|C=c} $
for $b\in {\cal B}_j$ and $c\in {\cal C}_j$. 
Therefore, conditioned on that $B = b \in {\cal B}_j$, we must have $C = c \in \mathcal C_j$ with probability one, and vice versa.

We define the random variable $J$ to indicate the label of subset ${\cal B}_j$ that
contains the variable $B$, i.e., put $J = j$ whenever $B = b\in \mathcal B_j$.
Following from this construction, once the variable $J$ is fixed to be $j$,
the state $\rho_{A|b,c}$ does not depend on the choice of $b,c$
in ${\cal B}_j\times {\cal C}_j$. This means $A$ is conditionally independent of the pair $(B,C)$ given the classical label $J$. Hence, we have $I(A;B,C|J)=0$. We denote the function $b\mapsto j$ and $c\mapsto j$ by $g_1$ and $g_2$.

\noindent{\bf Step 3:}
Next, we prove the condition (ii) when $B$ and $C$ are quantum systems. 
when $A-B-C$ and $A-C-B$ hold, i.e., the condition (i) holds.

Due to Theorem 6 of \cite{HJP},
${\cal H}_B$ and ${\cal H}_C$ are decomposed as
\begin{align}
{\cal H}_B =\bigoplus_k  {\cal H}_{B_1,k} \otimes {\cal H}_{B_2,k} ,~
{\cal H}_C =\bigoplus_l  {\cal H}_{C_1,l} \otimes {\cal H}_{C_2,l} ,
\end{align}
and the state $\rho_{ABC}$ has the block structures
\begin{align}
\rho_{ABC}&=\bigoplus_k \rho_{AB_1|k} \otimes \rho_{CB_2|k} \label{C1}\\
\rho_{ABC}&=\bigoplus_l \rho_{AC_1|l} \otimes \rho_{BC_2|l} ,\label{C2}
\end{align}
where $\rho_{AB_1|k} $, $\rho_{CB_2|k} $,
$\rho_{AC_1|l}$, $ \rho_{BC_2|l}$ are states on 
${\cal H}_A \otimes  {\cal H}_{B_1,k} $, \
${\cal H}_C\otimes {\cal H}_{B_2,k}$, \
${\cal H}_A\otimes {\cal H}_{C_1,l}$, \
${\cal H}_B \otimes {\cal H}_{C_2,l}$, and
$\oplus$ denotes orthogonal direct sums of Hilbert spaces.

Define projectors:
\begin{align}
  M_{B,k} &:= I_A \otimes P_{\mathcal{H}_{B_1,k} \otimes \mathcal{H}_{B_2,k}} \otimes I_C, \\
  M_{C,l} &:= I_A \otimes I_B \otimes P_{\mathcal{H}_{C_1,l} \otimes \mathcal{H}_{C_2,l}},
\end{align}
 $P_{\mathcal{H}}$ denotes the orthogonal projection onto the subspace $\mathcal{H}$.
We define the random variable 
$K$ and $L$
by the joint distribution 
$P_{K,L}(k,l):=\Tr (I_A \otimes M_{B,k}\otimes M_{C,l} \ \rho_{ABC})$.
Introduce an auxiliary classical-quantum state: 
$$\rho_{AKL}
:=\sum_{k,l}
(\Tr_{BC} I_A \otimes M_{B,k}\otimes M_{C,l} \rho_{ABC}) \otimes
|k,l\rangle \langle k,l|.$$
The relations \eqref{C1} and \eqref{C2} 
imply 
the Markov chains $A-K-L$ and $A-L-K$, respectively.

Now, we fix the variables $K$ and $L$ to be $k$ and $l$ and consider the state $\rho_{ABC|K=k,l=l}$
living in space ${\cal H}_A\otimes {\cal H}_{B_1,k} \otimes {\cal H}_{B_2,k}
\otimes {\cal H}_{C_1,l} \otimes {\cal H}_{C_2,l}$.
Let $\cdot |k$ denote the system restricted by the projector $M_{B,k}$. 
The original Markov chain condition implies 
the Markov chains $A|k-(B_1,B_2)|k-(C_1, C_2)|k$
and $A|l-(C_1, C_2)|l-(B_1,B_2)|l$.
The decomposition \eqref{C1} implies the relation $A|k,B_1|k \perp B_2|k,C_1|k, C_2|k$
and 
the form \eqref{C2} implies the relation $A|l,C_1|l \perp C_2|l,B_1|l, B_2|l$.
Thus, we find that
$A|(k,l)$, $B_1|(k,l)$, $C_1|(k,l)$, $(B_2,C_2)|(k,l)$ are independent of each other.
That is, 
the state $\rho_{ABC|K=k,L=l}$ has the form
$\rho_{A|(k,l)}\otimes \rho_{B_1|(k,l)}\otimes \rho_{C_1|(k,l)}
\otimes \rho_{B_2,C_2,k,l}$.

Apply the result from the case when $(B, C)$ are classical random variables to the system 
$AKL$, and define the random variable $J$ 
and the functions $g_1$ and $g_2$
in the way as before, i.e., $J:= g_1(K) = g_2(L)$. 
Therefore, the PVMs $\{E_{B,j}\}_j$ and $\{E_{C,j}\}_j$ satisfies \eqref{eq:block-diag}
while their each block with the index $j$ is given as
\begin{align}
  E_{B,j}& := \sum_{k: g_1(k)=j} P_{\mathcal{H}_{B_1,k} \otimes \mathcal{H}_{B_2,k}}, \\
  E_{C,j} &:= \sum_{l: g_2(l)=j} P_{\mathcal{H}_{C_1,l} \otimes \mathcal{H}_{C_2,l}}.
\end{align}

Now, conditioned on $J = j$, since $\rho_{A|k,l} = \rho_{A|j}$, we sum over all $(k,l)$ pairs such that $g_1(k) = g_2(l) = j$ to obtain:
\begin{align}
  &\rho_{ABC|J=j} 
  = \sum_{\substack{k,l: \\ g_1(k)=g_2(l)=j}} P_{K,L|J=j}(k,l) \rho_{ABC|K=k, L=l}  \notag \\
  =& \sum_{\substack{k,l: \\ g_1(k)=g_2(l)=j}} P_{K,L|J=j}(k,l)\, \Big(\rho_{A|j} \otimes \rho_{B_1|(k,l)}  \notag\\
 &\qquad \otimes \rho_{C_1|(k,l)} \otimes \rho_{B_2,C_2|(k,l)}\Big) \notag\\
  =& \rho_{A|j} \otimes \Bigg(\sum_{\substack{k,l: \\ g_1(k)=g_2(l)=j}} P_{K,L|J=j}(k,l)\, \Big( \rho_{B_1|(k,l)} \notag\\
   &\qquad \otimes \rho_{C_1|(k,l)} \otimes \rho_{B_2,C_2|(k,l)} \Big)\Bigg).
\end{align}
Hence, we conclude that
$A|j \perp (B_1,B_2,C_1,C_2)|j$, in other words, $I(A;BC|J)=0$, which implies the condition (ii).



\end{proof}

\section{Conditional double Markovity for quantum states with full support}\label{sec:quantumDM2}
Theorem \ref{TH1} is generalized 
with the additional system $D$ as follows
when the full support condition is imposed.

\begin{theorem}\label{TH2}
Let $A,B,C,D$ be quantum systems with \textbf{strictly positive} joint state $\rho_{ABCD}$. The following two Markov structures are equivalent:
\begin{enumerate}[(i)]
    \item $A-(B,D)-C$ and $A-(C,D)-B$;
    \item $A-D-(B,C)$.
\end{enumerate}
\end{theorem}

\subsection{Minimal decomposition}\label{SG}
The direct-sum/tensor-product decomposition used throughout this paper is the one appearing in
the structure theorem for quantum Markov states in \cite{HJP}.
For our purposes, we introduce a \emph{minimal} version of such a decomposition tailored to a
single strictly positive bipartite state $\rho_{XY}$.
In particular, we formulate an explicit minimality criterion and prove a uniqueness property
under full support, which will be a key ingredient in the proof of Theorem~\ref{TH2}.
The system ${\cal H}_Y$ can be decomposed as
\begin{align}
{\cal H}_Y &=\bigoplus_k  {\cal H}_{Y_1|k} \otimes {\cal H}_{Y_2|k} \label{TS1}.
\end{align}
The state $\rho_{XY}$ has the form
\begin{align}
\rho_{XY}&=\bigoplus_k \rho_{X,Y_1|k} \otimes \rho_{Y_2|k} \label{TS2}
\end{align}
where $\rho_{X,Y_1|k} $, $\rho_{Y_2|k} $ are 
strictly positive matrices on 
${\cal H}_X\otimes {\cal H}_{Y_1|k}$ and $ {\cal H}_{Y_2|k}$.

We say that the above decomposition \eqref{TS1}
is minimal with respect to $X$ and $\rho_{XY}$
when the following conditions hold.
\begin{enumerate}[(i)]
    \item There is no strict subspace of ${\cal H}_{Y_1|k}$ that can be served as a component of a decomposition.
    \item There is no pair $k,k'$ such that 
$\rho_{X,Y_1|k}$ and $\rho_{X,Y_1|k'}$ are related via a unitary transformation.
\end{enumerate}
A minimal decomposition exists because
\begin{enumerate}[(i)]
    \item If $\rho_{X,Y_1|k}$ can be written as direct sum $\bigoplus_l \rho_{X,Y_1|k,l}$, we obtain a refinement by replacing the original block by direct sum of smaller blocks $\oplus_l \rho_{X,Y_1|k,l} \otimes \rho_{Y_2|k} $.
    \item If there is a pair $k,k'$ in the sense of (2), combining ${\cal H}_{Y2|k}$ and 
${\cal H}_{Y2|k'}$ via direct sum removes such a pair.
\end{enumerate}

In general, a minimal decomposition is not unique. If the support of $\rho_{Y_2|k} $ is strictly smaller than $ {\cal H}_{Y_2|k}$,
we focus on the kernel $\ker(\rho_{Y_2|k} )$, which is its orthogonal subspace to the support $\operatorname{supp} (\rho_{Y2|k} )$
in ${\cal H}_{Y2|k}$.

The subspace $\mathcal{H}_{Y_1|k} \otimes \ker(\rho_{Y_2|k})$ carries no weight in the state $\rho_{XY}$, so it could be reassigned to a different block $k'$, without harming the decomposition. That is:
\[
  \mathcal{H}_{Y_1|k} \otimes \ker(\rho_{Y_2|k}) \cong \mathcal{H}_{Y_1|k'} \otimes \ker(\rho_{Y_2|k})
\]
for any $k' \neq k$, without changing the state. This ambiguity shows that when $\rho_{Y_2|k}$ is not full rank, the minimal decomposition is not unique.



In contrast, when the state $\rho_{XY}$ on the joint system $XY$
has full support, there uniquely exists a minimal decomposition with respect to $X$ and 
a strictly positive density matrix $\rho_{XY}$.

\begin{lemma}
    Let $\rho_{XY}$ be a strictly positive joint state. Suppose $\rho_{XY}$ admits two minimal decompositions:
\begin{align}
  \mathcal{H}_Y &= \bigoplus_k \mathcal{H}_{Y_1|k} \otimes \mathcal{H}_{Y_2|k}, \quad
  \rho_{XY} = \bigoplus_k \rho_{X,Y_1|k} \otimes \rho_{Y_2|k}, \label{eq:decomp1}\\
  \mathcal{H}_Y &= \bigoplus_{k'} \mathcal{H}'_{Y_1|k'} \otimes \mathcal{H}'_{Y_2|k'}, \quad
  \rho_{XY} = \bigoplus_{k'} \rho'_{X,Y_1|k'} \otimes \rho'_{Y_2|k'}. \label{eq:decomp2}
\end{align}
Then the two decompositions are equivalent up to reordering and unitary equivalence. That is, there exists a bijection $\phi: \{k\} \to \{k'\}$ and unitaries $U_k$ such that
\[
 U_k \rho_{X,Y_1|k} \otimes \rho_{Y_2|k}U^\dagger_k = \rho'_{X,Y_1|\phi(k)} \otimes \rho'_{Y_2|\phi(k)}.
\]
\end{lemma}

\begin{proof}
    Denote subspaces $\mathcal{H}_{Y_1|k} \otimes \mathcal{H}_{Y_2|k} $ by $\mathcal V_k$ and $ \mathcal{H}'_{Y_1|k'} \otimes \mathcal{H}'_{Y_2|k'}$ by $\mathcal V'_{k'}$, and let $\mathcal W_{k,k'} = \mathcal V_k\cap \mathcal V'_{k'}$. 

    Consider a fixed $k$, we wish to show that there exists only one $\mathcal W_{k,k'}$ that is non-trivial and satisfies $\mathcal V_k = \mathcal W_{k,k'}$. Assume the contrast, there are multiple $k'$'s, say $\{k'_1,k'_2,\ldots, k'_m\}$, $m\geq 2$, such that $\mathcal W_{k,k'}$ is non-trivial. Let $\rho'_{X,Y_1|k,k'}$ denote the state $\rho'_{X,Y_1|k'}$ restricted on the space $\mathcal W_{k,k'}$. Our assumption yields:
    \begin{align*}
        \rho_{X,Y_1|k}\otimes \rho_{Y_2|k} &= \bigoplus_{i=1}^m  \rho'_{X,Y_1|k,k'_i} \otimes \rho'_{Y_2|k,k'_i}.
    \end{align*}
    Since $\rho_{XY}$ is strictly positive, thus $\rho'_{X,Y_1|k,k'_i}\otimes \rho'_{Y_2|k,k'_i}>0$ as $\mathcal W_{k,k'_i}$ is non-trivial subspace. Therefore, we obtain a refinement of the decomposition (\ref{eq:decomp1}). Contradicting the minimality of the decomposition. 

    The minimality condition also guarantees that $\rho_{X,Y_1|k}$ and $\rho_{X,Y_1|l}$ are not unitarily equivalent for $k\neq l$. Hence, there must be an one-to-one correspondence $\phi$ from $\{k\}$ to $\{k'\}$, and $\rho_{X,Y_1|k}$ is unitarily equivalent to $\rho'_{X,Y_1|\phi(k)}$.
\end{proof}

\subsection{Proof of Theorem \ref{TH2}}\label{SH}
\begin{proof}
In the proof, we always choose the HJP direct-sum decomposition to be minimal in the sense of Section~\ref{SG}. 
Since $\rho_{ABCD}>0$, the minimal decomposition with respect to AAA is unique up to permutation and local unitaries. This allows us to identify the decompositions obtained from $A-(B,D)-C$ and from 
$A-(C,D)-B$ blockwise.

Again, $(ii)\implies (i)$ is immediate because $I(A;C|B,D)\leq I(A;B,C|D) = 0$ and $I(A;B|C,D)\leq I(A;B,C|D) = 0$. 

Consider the reverse direction. We assume the Markov chains $A-(B,D)-C$ and $A-(C,D)-B$ hold.
By applying 
Theorem 6 of \cite{HJP} to the Markov chain $A-(B,D)-C$,
${\cal H}_{B}\otimes {\cal H}_{D}$ is decomposed as
\begin{align}
{\cal H}_{B}\otimes {\cal H}_{D} =\bigoplus_k  {\cal H}_{BD_1,k} \otimes {\cal H}_{BD_2,k} \label{CS1}.
\end{align}
The state $\rho_{ABCD}$ has the form
\begin{align}
\rho_{ABCD}&=\bigoplus_k \rho_{A,BD_1|k} \otimes \rho_{C,BD_2|k} \label{CS2}
\end{align}
where $\rho_{A,BD_1|k} $, $\rho_{C,BD_2|k} $ are 
strictly positive matrices on 
${\cal H}_A \otimes {\cal H}_{BD_1,k}$ and ${\cal H}_C \otimes {\cal H}_{BD_2,k}$.

Since the minimality condition with respect to $A$ and $\rho_{ABCD}$
is decided by the structures of ${\cal H}_{BD1,k}$,
when the decomposition \eqref{CS1} 
is minimal with respect to 
$A$ and $\rho_{ABD}$,
the decomposition
\begin{align}
{\cal H}_{B}\otimes{\cal H}_{C}\otimes {\cal H}_{D} =\bigoplus_k  {\cal H}_{BD_1,k} \otimes 
{\cal H}_C \otimes {\cal H}_{BD_2,k} \label{CS3}
\end{align}
is minimal with respect to $A$ and $\rho_{ABCD}$.
Now, we assume that the decomposition \eqref{CS1} 
is minimal with respect to $A$ and $\rho_{ABD}$.

Next, similarly, by applying 
Theorem 6 of \cite{HJP} to the Markov chain $A-(C,D)-B$,
${\cal H}_{C}\otimes {\cal H}_{D}$ is decomposed as
\begin{align}
{\cal H}_{C}\otimes {\cal H}_{D} =\bigoplus_l  {\cal H}_{CD_1,l} \otimes {\cal H}_{CD_2,l} \label{CS4}.
\end{align}
The state $\rho_{ABCD}$ has the form
\begin{align}
\rho_{ABCD}&=\bigoplus_l \rho_{A,CD_1|l} \otimes \rho_{B,CD_2|l} \label{CS5}
\end{align}
where $\rho_{A,CD_1|l} $, $\rho_{B,CD_2|l} $ are 
strictly positive matrices on 
${\cal H}_A \otimes {\cal H}_{CD_1,l}$ and ${\cal H}_B \otimes {\cal H}_{CD_2,l}$.
Also, we assume that the decomposition \eqref{CS4} 
is minimal with respect to $A$ and $\rho_{ACD}$.
Hence, 
\begin{align}
{\cal H}_{B}\otimes {\cal H}_{C}\otimes {\cal H}_{D} 
=\bigoplus_l  {\cal H}_{CD_1,l} \otimes{\cal H}_{B}\otimes {\cal H}_{CD_2,l} \label{CS6}
\end{align}
is also minimal with respect to $A$ and $\rho_{ABCD}$.
Since $\rho_{ABCD}$ is a strictly positive density matrix,
the uniqueness of the minimal decomposition with respect to $A$ and $\rho_{ABCD}$
guarantee the one-to-one map $k \mapsto l(k)$ such that
\begin{align*}
    {\cal H}_{BD_1,k}= {\cal H}_{CD_1,l(k)}, \quad
    {\cal H}_{C}\otimes{\cal H}_{BD_2,k}={\cal H}_{B}\otimes {\cal H}_{CD_2,l(k)},
\end{align*}
with
\begin{align*}
    \rho_{A,BD_1|k} = \rho_{A,CD_1|l(k)},\quad 
    \rho_{C,BD_2|k} =\rho_{B,CD_2|l(k)}.
\end{align*}


Therefore, for each $k$, there exists a subspace 
${\cal K}_{D_1,k}\otimes {\cal K}_{D_2,k}\subset {\cal H}_D$ such that
${\cal H}_{C}\otimes{\cal H}_{BD_2,k}={\cal H}_{B}\otimes {\cal H}_{CD_2,l(k)}$
restricted on $\mathcal H_D$ gives ${\cal K}_{D_2,k}$
and 
${\cal H}_{BD_1,k}= {\cal H}_{CD_1,l(k)}$ restricted on $\mathcal H_D$ gives ${\cal K}_{D_1,k}$. 
Since the $k$-th subspace of decomposition (\ref{CS3}) can be written as
\begin{align*}
    {\cal H}_{BD_1,k} \otimes {\cal H}_C \otimes {\cal H}_{BD_2,k} &= {\cal H}_{BD_1,k} \otimes {\cal H}_B \otimes {\cal H}_{CD_2,l(k)},
\end{align*}
yielding ${\cal H}_C \otimes {\cal H}_{BD_2,k}$ carries the whole $\mathcal H_B\otimes \mathcal H_C$, while $\mathcal H_{BD_1,k}$ does contain any part of $\mathcal H_B$, i.e., $\mathcal H_{BD_1,k} = \mathcal K_{D_1,k}$. 

Therefore, the space ${\cal H}_D$ is decomposed as
\begin{align}
{\cal H}_D=\oplus_k {\cal K}_{D1,k}\otimes {\cal K}_{D2,k}.
\end{align}
the state $\rho_{A,BD_1|k} = \rho_{A,CD_1|l(k)}$ 
can be written as a state $\rho_{A,D_1|k}$
on ${\cal H}_A\otimes {\cal K}_{D_1,k}$, 
and 
$\rho_{C,BD_2|k} =\rho_{B,CD_2|l(k)}$
can written as a state $\rho_{BC,D_2|k}$
on ${\cal H}_B\otimes {\cal H}_C\otimes {\cal K}_{D_2,k}$.
Hence, the state $\rho_{ABCD}$ has the block form
\begin{align}
\rho_{ABCD}=\oplus_k \rho_{A,D_1|k} \otimes \rho_{BC,D_2|k},
\end{align}
which implies the Markov chain $A-D-(B,C)$.
\end{proof}

\section{Conclusion}
In this paper, we have established quantum versions of the double Markovity phenomenon and have provided two structural characterizations that parallel the most commonly used classical formulations.
First, for a tripartite state $\rho_{ABC}$, we have shown that the simultaneous Markov conditions
$A\!-\!B\!-\!C$ and $A\!-\!C\!-\!B$ are equivalent to the existence of projective measurements on $B$ and $C$
that induce a common classical label $J$ and render $A\!-\!J\!-\!(B,C)$ a Markov chain (Theorem~\ref{TH1}).
This characterization has identified an explicitly classical piece of information---the index $J$---that captures the part of $B$ and $C$ relevant to screening off $A$, thereby providing a faithful quantum analogue of the classical ``common summary variable'' picture.

Second, under a full-support (strictly positive) assumption, we have established a conditional double Markovity statement for four-party states:
for $\rho_{ABCD}>0$, the pair of Markov relations $A\!-\!(B,D)\!-\!C$ and $A\!-\!(C,D)\!-\!B$
has been shown to be equivalent to $A\!-\!D\!-\!(B,C)$ (Theorem~\ref{TH2}).
The proof has relied on a refined use of the structural decomposition of quantum Markov states, together with a uniqueness property of a minimal direct-sum decomposition under full support.

Taken together, these results have removed a key bottleneck in extending the subadditivity--doubling--rotation methodology to quantum settings.
In particular, they have supplied the missing equality-case mechanism that is frequently required when strict subadditivity is invoked in SDR-type arguments.
We expect that the present structural theorems will serve as a reusable tool in future work on quantum Gaussian extremality and related information inequalities, where one needs to convert Markov-type conclusions into explicit structural constraints on quantum states.

\bibliographystyle{IEEEtran}
\bibliography{mybiblio}

\end{document}

%% file: Cheadmin.tex
\usepackage{amsmath,amssymb,amsthm}
\usepackage{tikz,enumerate}
\usepackage{cite,mathtools}
\usepackage{etex}
\usepackage[normalem]{ulem}
\theoremstyle{plain}

\newtheorem{theorem}{Theorem}

\newtheorem{lemma}{Lemma}

\newtheorem*{lemma*}{Lemma}

\theoremstyle{remark}

\theoremstyle{definition}

\DeclareMathOperator{\Tr}{Tr}